\documentclass[submission,copyright,creativecommons]{eptcs}
\providecommand{\event}{NCMA 2023} %

\usepackage{iftex}

\ifpdf
  \usepackage{underscore}         %
  \usepackage[T1]{fontenc}        %
\else
  \usepackage{breakurl}           %
\fi

\usepackage{amsmath}
\usepackage{amssymb}
\usepackage{tikz}
\usetikzlibrary{automata,arrows}
\tikzset{
	state/.append style={semithick,initial text=,inner sep=0},
	transition/.append style={->,>=stealth',shorten >=1pt,semithick},
	transition label/.append style={outer sep=3pt},
	every initial by arrow/.append style={initial text=,initial where=,transition},
}

\usepackage{xspace}

\newcommand{\set}[1]{\mbox{$\{#1\}$}}
\newcommand{\IN}{\mathbb N}
\newcommand{\IZ}{{\mathbb Z}}

\DeclareMathOperator{\lcm}{lcm}

\newcommand*{\mA}{\ensuremath{A}\xspace}

\newcommand\s{s\xspace}
\newcommand\ow{\textsc1\xspace}
\newcommand\tw{\textsc2\xspace}

\newcommand\om{\textsc{om}\xspace}
\newcommand\am{\textsc{am}\xspace}
\newcommand\f{\textsc{f}\xspace}

\newcommand\dfa{\textsc{dfa}\xspace}
\newcommand\nfa{\textsc{nfa}\xspace}
\newcommand\twdfa{{\tw}{\dfa}\xspace}
\newcommand\twnfa{{\tw}{\nfa}\xspace}
\newcommand\owdfa{{\ow}{\dfa}\xspace}
\newcommand\ownfa{{\ow}{\nfa}\xspace}

\newcommand\twdfas{{\twdfa}\s}
\newcommand\twnfas{{\twnfa}\s}
\newcommand\owdfas{{\owdfa}\s}
\newcommand\ownfas{{\ownfa}\s}

\newcommand\flas{{\fla}\s}  %
\newcommand\dflas{{\dfla}\s} %

\newcommand*\la[1][1]{\ensuremath{#1}\textsc{\ifx&#1&\else-\fi la}\xspace}
\newcommand*\dla[1][1]{\textsc{d-}\la[#1]}%

\newcommand*\omla[1][1]{\om-\ensuremath{#1}\textsc{\ifx&#1&\else-\fi la}\xspace}

\newcommand*\amla[1][1]{\am-\ensuremath{#1}\textsc{\ifx&#1&\else-\fi la}\xspace}

\newcommand*\las[1][1]{{\la[#1]}\s}   %
\newcommand*\dlas[1][1]{{\dla[#1]}\s} %

\newcommand*\fla[1][1]{\f-\ensuremath{#1}\textsc{\ifx&#1&\else-\fi la}\xspace}
\newcommand*\dfla[1][1]{\textsc{d-}\fla[#1]}%

\newcommand*\lend{{\ensuremath{\mathord{\vartriangleright}}}\xspace}
\newcommand*\rend{{\ensuremath{\mathord{\vartriangleleft}}}\xspace}

\newcommand\bigoof[1]{\ensuremath{O\left(#1\right)}}

\newcommand*{\qed}{\mbox{}\nolinebreak\hfill~\raisebox{0.77ex}[0ex]{\framebox[1ex][l]{}}}

\newtheorem{theorem}{Theorem}
\newtheorem{lemma}{Lemma}
\newtheorem{corollary}{Corollary}

\newenvironment{proof}{\noindent{\em Proof.}}{\bigskip\noindent}
\newtheorem{ex}{Example}

\newenvironment{example-cont}[1]{\bigskip\noindent\textbf{Example~\ref{#1}.~(cont.)\hspace{\labelsep}}}{\bigskip\noindent}

\usepackage[capitalize,noabbrev,sort&compress,nameinlink]{cleveref}

\title{Forgetting $1$-Limited Automata}
\author{Giovanni Pighizzini
\institute{Dipartimento di Informatica\\
Universit\`{a} degli Studi di Milano, Italy}
\email{pighizzini@di.unimi.it}
\and
Luca Prigioniero
\institute{Department of Computer Science\\
Loughborough University, UK}
\email{l.prigioniero@lboro.ac.uk}
}

\def\titlerunning{Forgetting $1$-Limited Automata}
\def\authorrunning{G.~Pighizzini \& L.~Prigioniero}

\begin{document}
\maketitle

\begin{abstract}
\noindent
We introduce and investigate \emph{forgetting $1$-limited automata}, which are single-tape Turing machines that,
when visiting a cell for the first time,
replace the input symbol in it by a fixed symbol, so forgetting the original contents.
These devices have the same computational power as finite automata, namely they characterize the class of regular languages.
We study the cost in size of the conversions of forgetting $1$-limited automata, in both nondeterministic and deterministic cases,
into equivalent  one-way nondeterministic and deterministic automata, providing optimal bounds in terms of exponential or
superpolynomial functions. We also discuss the size relationships with two-way finite automata. In this respect,
we prove the existence of a language for which forgetting $1$-limited automata are exponentially larger than
equivalent minimal deterministic two-way automata.
\end{abstract}

\section{Introduction}
\label{sec:intro}

Limited automata have been introduced in 1967 by Hibbard, with the aim of generalizing the notion of
determinism for context-free languages~\cite{Hi67}. These devices regained attention in the last decade,
mainly from a descriptional complexity point of view, and they
have been considered in several papers, starting with~\cite{PP14,PP15}.
(For a recent survey see~\cite{Pig19}.)

In particular, \emph{$1$-limited automata} are single-tape nondeterministic Turing machines that are allowed to rewrite 
the content of each tape cell only in the first visit. They have the same computational power as finite 
automata~\cite[Thm.~12.1]{WW86}, but they can be extremely more succinct.
Indeed, in the worst case the size gap from the descriptions of $1$-limited automata to those of equivalent one-way
deterministic finite automata is double exponential~\cite{PP14}.

In order to understand this phenomenon better, we recently studied two restrictions of $1$-limited automata~\cite{PP23b}.
In the first restriction, called \emph{once-marking $1$-limited automata}, during each computation the machine can make only one
change to the tape, just marking exactly one cell during the first visit to it. We proved that, under this restriction, a double exponential
size gap to one-way deterministic finite automata remains possible.

In the second restriction, called \emph{always-marking $1$-limited automata}, each tape cell is marked during the first visit.
In this way, at each step of the computation, the original content in the cell remains available, together with the information
saying if it has been already visited at least one time. In this case, the size gap to one-way deterministic finite automata reduces
to a single exponential. However, the information about which cells have been already visited still gives extra 
descriptional power. In fact, the conversion into equivalent two-way finite automata in the worst case costs exponential in size,
even if the original machine is deterministic and the target machine is allowed to make nondeterministic choices.

A natural way to continue these investigations is to ask what happens if in each cell the information about the original input
symbol is lost after the first visit. This leads us to introduce and study the subject of this paper, namely
\emph{forgetting $1$-limited automata}. These devices are $1$-limited automata in which, during the first visit to
a cell, the input symbol in it is replaced with a unique fixed symbol. Forgetting automata have been introduced in the
literature longtime ago~\cite{JMP93}. Similarly to the devices we consider here, they can use only one fixed symbol
to replace symbols on the tape. However, the replacement is not required to happen in the first visit, 
so giving the possibility to recognize more than regular languages.
In contrast, being a restriction of~$1$-limited automata, forgetting $1$-limited automata recognize only regular languages.

In this paper, first we study the size costs of the simulations of forgetting $1$-limited automata, in both nondeterministic
and deterministic versions, by one-way finite automata. 
The upper bounds we prove are exponential,
when the simulated and the target machines are nondeterministic and deterministic, respectively. In the other cases they are superpolynomial. These bounds are obtained starting from the conversions of
always-marking $1$-limited automata into one-way finite automata presented in~\cite{PP23b}, whose costs, in the case we are considering,
can be 
reduced using techniques and results derived in the context of automata over a one-letter alphabet~\cite{Ch86,MP01}.
We also provide witness languages showing that these upper bounds cannot be improved asymptotically.

In the last part of the paper we discuss the relationships with the size of two-way finite automata, which are not completely clear.
We show that losing the information on the input
content can reduce the descriptional power. In fact, we show languages for which forgetting $1$-limited automata, even
if nondeterministic, are exponentially larger than minimal two-way deterministic finite automata.
We conjecture that also the converse can happen. In particular we show a family of languages for which we conjecture that two-way
finite automata, even if nondeterministic, must be significantly larger than minimal deterministic forgetting $1$-limited
automata.

\section{Preliminaries}
\label{sec:prel}

In this section we recall some basic definitions useful in the paper.
Given a set~$S$,
$\#{S}$~denotes its cardinality and~$2^S$ the family of all its subsets.
Given an alphabet~$\Sigma$ and a string~$w\in\Sigma^*$,
$|w|$ denotes the length of~$w$,
$|w|_a$ the number of occurrences of~$a$ in~$w$, and
$\Sigma^k$ the set of all strings on~$\Sigma$ of length~$k$.

We assume the reader to be familiar with notions from formal languages and automata
theory, in particular with the fundamental variants of finite automata
(\owdfas, \ownfas, \twdfas, \twnfas, for short, where \ow/\tw mean
\emph{one-way}/\emph{two-way} and \textsc{d}/\textsc{n} mean
\emph{deterministic}/\emph{nondeterministic}, respectively).
For any unfamiliar terminology see, e.g.,~\cite{HU79}.

A \emph{$1$-limited automaton} (\la, for short)
is a tuple $\mA=(Q,\Sigma,\Gamma,\delta,q_I,F)$,
where~$Q$ is a finite \emph{set of states},
$\Sigma$ is a finite \emph{input alphabet},
$\Gamma$ is a finite \emph{work alphabet} such that~$\Sigma \cup
\{\lend,\rend\} \subseteq \Gamma$,
$\lend,\rend \notin \Sigma$
are two special symbols,
called the \emph{left} and the \emph{right end-markers},
$\delta:Q\times\Gamma\rightarrow 2^{Q\times(\Gamma\setminus\{\lend,\rend\})\times\{-1,+1\}}$ is the \emph{transition function},
and~$F\subseteq Q$ is a set of final states.
At the beginning of the computation,
the input word~$w\in\Sigma^*$ is stored onto the tape surrounded by the two end-markers,
the left end-marker being in position zero
and
the right end-marker being in position~$|w|+1$.
The head of the automaton is on cell~$1$ and the state of the finite control is the \emph{initial state}~$q_I$.

In one move,
according to~$\delta$ and the current state,
\mA reads a symbol from the tape,
changes its state,
replaces the symbol just read from the tape with a new symbol,
and moves its head to one position forward or backward.
Furthermore, the head cannot pass the end-markers,
except at the end of computation,
to accept the input, as explained below.
Replacing symbols is allowed to modify the content of each cell only
during the first visit,
with the exception of the cells containing the end-markers,
which are never modified.
Hence, after the first visit, a tape cell is ``frozen''.
More technical details can be found in~\cite{PP14}.

The automaton \mA accepts an input~$w$ if and only if there is a computation path that
starts from the initial state~$q_I$ with the input tape containing~$w$
surrounded by the two end-markers and the head on the first input cell, and
which ends in a \emph{final state}~$q\in F$ after passing the right
end-marker.
The device \mA~is said to be \emph{deterministic} (\dla, for short) whenever~$\#{\delta(q,\sigma)}\le 1$, for every $q\in Q$ and $\sigma\in\Gamma$.

We say that the \la \mA is a \emph{forgetting \la} (for short \fla or \dfla in the deterministic case), when
there is only one symbol~$Z$ that is used to replace symbols in the first visit, i.e., the work alphabet 
is~$\Gamma=\Sigma \cup \{Z\}\cup\{\lend,\rend\}$, with~$Z\notin\Sigma$ and
if~$(q,A,d)\in\delta(p,a)$ and~$a\in\Sigma$ then~$A=Z$.

Two-way finite automata  are limited automata in which no rewritings are possible;
one-way finite automata can scan the input in a one-way fashion only.
A finite automaton is, as usual, a tuple~$(Q,\Sigma,\delta,q_I,F)$,
where,
analogously to \las,
~$Q$ is the finite set of states,
$\Sigma$ is the finite input alphabet,
$\delta$ is the transition function,
$q_I$ is the initial state,
and~$F$ is the set of final states.
We point out that for two-way finite automata we assume the same accepting conditions as for \las. 

Two-way machines 
in which the direction of the head can change only at the end-markers
are said to be \emph{sweeping}~\cite{Sip80b}.

\medskip

In this paper we are interested in comparing the size of machines.
The \emph{size} of a model
is given by the total number of symbols used to write down its description.
Therefore,
the size of \las is bounded by a polynomial
in the number of states and of work symbols,
while,
in the case of finite automata,
since no writings are allowed,
the size is linear in the number of instructions and states,
which is bounded by a polynomial in the number of states
and in the number of input symbols.
We point out that, since~\flas use work alphabet~$\Gamma=\Sigma \cup \set{Z}\cup\{\lend,\rend\}$, $Z \notin \Sigma$,
the relevant parameter for evaluating the size of these devices is their number of states,
differently than \las,
in which the size of the work alphabet is not fixed,
i.e., depends on the machine.

We now shortly recall some notions and results related to number theory that will be useful to obtain our cost estimations.
First, given two integers~$m$ and $n$, let us denote by~$\gcd(m,n)$ and by~$\lcm(m,n)$ their \emph{greatest common divisor} and
\emph{least common multiple}, respectively.

We remind the reader that each integer~$\ell>1$ can be factorized in a unique way as
product of powers of primes, i.e., as~$\ell=p_1^{k_1}\cdots p_r^{k_r}$,
where~$p_1<  \cdots < p_r$ are primes, and $k_1,\ldots,k_r > 0$.

In our estimations, we shall make use of the \emph{Landau's
function}~$F(n)$~\cite{La03,La09}, which plays an important role in the analysis of 
simulations among different types of unary automata~(e.g. \cite{Ch86,Ge07,MP01}).
Given a positive integer~$n$, let
\[
  F(n) = \max\{\lcm(\lambda_1,\ldots,\lambda_r)\;\mid\;
    \lambda_1+\cdots+\lambda_r=n\}\,,
\]
where $\lambda_1,\ldots,\lambda_r$ denote, for the time
being, arbitrary positive integers. Szalay~\cite{Sz80} gave
a sharp estimation of~$F(n)$ that, after some simplifications, can
be formulated as follows:
\[
  F(n)=e^{(1+o(1))\cdot\sqrt{n\cdot\ln n}}.
\label{e:estF}
\]
Note that the function~$F(n)$ grows less than~$e^n$, but more than each polynomial
in~$n$. In this sense we say that~$F(n)$ is a \emph{superpolynomial function}.

As observed in~\cite{GP12}, for each integer~$n>1$ the value of~$F(n)$ can also be
expressed as the maximum product of powers of primes, whose sum is bounded by~$n$,
i.e.,
\[
  F(n) =
    \max\{ p_1^{k_1}\cdots p_r^{k_r}
    \;\mid\; p_1^{k_1}+\cdots+p_r^{k_r}\leq n\text,\; p_1,\ldots,p_r  
    \mbox{ are primes, }
    \mbox{and } k_1,\ldots,k_r > 0 \}\text.
\]

\section{Forgetting $1$-Limited Automata vs.\ One-Way Automata}
\label{sec:toOneWay}

When forgetting 1-limited automata visit a cell for the first time,
they replace the symbol in it with a fixed symbol~$Z$,
namely they forget the original content.
In this way, each input prefix can be rewritten in a unique way.
As already proved for \emph{always-marking} \las, this prevents a double exponential
size gap in the conversion to \owdfas~\cite{PP23b}. However, in this case the upper bounds obtained for always-marking \las, can be
further reduced, using the fact that only one symbol is used to replace input symbols:

\begin{theorem}
\label{th:upperBoundF-LA}
  Let~$M$ be an~$n$-state \fla.
  Then $M$ can be simulated by a~\ownfa with at most~$n\cdot (5n^2+F(n))+1$ states
  and by a complete \owdfa with at most~$(2^n-1)\cdot(5n^2+F(n))+2$ states.
\end{theorem}
\begin{proof}
	First of all, we recall the argument for the conversion of \las into~\ownfas and \owdfas\
	presented~\cite[Thm.~2]{PP14} that, in turn, is derived from the technique to convert \twdfas into equivalent 
	\owdfas, presented in~\cite{She59}, and based on \emph{transitions tables}.
	
	Let us start by supposing that~$M=(Q,\Sigma,\Gamma,\delta,q_0,F)$ is an~$n$-state \la.

	Roughly, transition tables represent the possible behaviors of~$M$ on ``frozen'' tape segments.
	More precisely, given~$z\in\Gamma^*$ ,
	the \emph{transition table} associated with~$z$ is the binary
	relation~$\tau_z\subseteq Q\times Q$, consisting of all pairs~$(p,q)$ such that~$M$ has a computation path that
	starts in the state~$p$ on the rightmost symbol of a tape segment containing $\lend z$,
	ends reaching the state~$q$ by leaving the same tape segment to the right side, i.e.,
	by moving from the rightmost cell of the segment to the right, and
	does not visit any cell outside the segment.
	
	A \ownfa~$A$ can simulate~$M$ by keeping in the finite control two components:
  	\begin{itemize}
  	\item The transition table corresponding to the part of the tape at the left of the head.
  		This part has been already visited and, hence, it is frozen.
  	\item The state in which the simulated computation of~$M$ reaches the current tape position.
	\end{itemize}
	Since the number of transition tables is at most~$2^{n^2}$, the number of states in the resulting \ownfa~$A$ is bounded by~$n\cdot 2^{n^2}$.
  
	Applying the subset construction, this automaton can be converted into an equivalent deterministic
	one, with an exponential increasing in the number of states, so obtaining a double exponential 
  	number of states in~$n$.
	In the general case, this number cannot be reduced due to the fact that different computations of~$A$,
	after reading the same input, could keep in the control different transition tables, depending on the fact that~$M$
	could replace the same input by different strings.
	
	We now suppose that~$M$ is a \fla. In this case each input string can be replaced by a unique string.
	This would reduce the cost of the conversion to \owdfas to a single exponential.
	Indeed, it is possible to convert the \ownfa~$A$ obtained from~$M$ into an equivalent \owdfa that keeps
	in its finite control the \emph{unique} transition table for the part of the tape scanned so far
	(namely, the same first component as in the state of~$A$),
	and the set of states that are reachable by~$M$ when entering the current tape cell (namely,
	a set of states that can appear in the second component of~$A$, while entering the
	current tape cell). This leads to an upper bound of~$2^n\cdot 2^{n^2}$ states for the resulting \owdfa.
	We can make a further improvement, reducing the number of transition tables
	used during the simulation.
	Indeed we are going to prove that only a subset of all the possible~$2^{n^2}$ transition tables can appear during the simulation.
	
	Since only a fixed symbol~$Z$ is used to replace input symbols on the tape,
	the transition table when the head is in a cell depends only on the position of the cell and not on the initial tape content. 
	
	For each integer~$m\geq 0$, 
	let us call~$\tau_m$ the transition table corresponding to a frozen tape segment of length~$m$, namely
	the transition table when the head of the simulating  one-way automaton is on the tape cell~$m+1$.
	We are going to prove that the sequence~$\tau_0,\tau_1,\ldots,\tau_m,\ldots$ is ultimately periodic,
	with period length bounded by~$F(n)$ and, more precisely, $\tau_m=\tau_{m+F(n)}$ for each~$m>5n^2$.
	
	The proof is based on the analysis of computation paths in unary \twnfas carried on in~\cite[Section~3]{MP01}.
	Indeed, we can see the parts of the computation on a frozen tape segment as computation paths of a unary \twnfa.
	More precisely, by definition, for~$p,q\in Q$, $\tau_m(p,q)=1$ if and only if there is a computation
	path~$C$ that enters the frozen tape segment of length~$m$ from the right
	in the state~$p$ and, after some steps, exits the segment to the right in the state~$q$.
	Hence, during the path~$C$ the head can visit only frozen cells (i.e., the cells in positions~$1,\ldots,m$)
	of the tape, and the left end-marker.
	There are two possible cases:
	\begin{itemize}
	\item\emph{In the computation path~$C$ the head never visits the left end-marker.}\\
	A path of this kind is also called \emph{left U-turn}. Since it does not depend on the position of the left end-marker,
	this path will also be possible, suitably shifted to the right, on each frozen segment of length~$m'>m$.
	Hence~$\tau_{m'}(p,q)=1$ for each~$m'\geq m$.
	Furthermore, it has been proven that if there is a left U-turn which starts in the state~$p$ on cell~$m$, 
	and ends in state~$q$, then there exists another left U-turn satisfying the same constraints, in which the
	head never moves farther than~$n^2$ positions to the left of the position~$m$~\cite[Lemma~3.1]{MP01}. 
	So, such a ``short'' U-turn can be shifted to the left, provided that the tape segment is longer than~$n^2$.
	
	Hence, in this case~$\tau_{m}(p,q)=1$ implies~$\tau_{m'}(p,q)=1$ for each~$m'>n^2$.
	
	\item\emph{In the computation path~$C$ the head visits at least one time the left end-marker.}\\ 
	Let~$s_0,s_1,\ldots,s_k$, $k\geq 0$, be the sequence of the states in which~$C$ visits the left
	end-marker. We can decompose~$C$ in a sequence of computation paths~$C_0, C_1,\ldots, C_k,C_{k+1}$, where:
	\begin{itemize}
		\item $C_0$ starts from the state~$p$ with the head on the cell~$m$ and ends in~$s_0$ when the head reaches the left 
		end-marker. $C_0$ is called \emph{right-to-left traversal} of the frozen segment.
		\item For~$i=1,\ldots,k$, $C_i$ starts in state~$s_{i-1}$ with the head on the left end-marker and
		ends in~$s_i$, when the head is back to the left end-marker. $C_i$ is called \emph{right U-turn}.
		Since, as seen before for left U-turns, each right U-turn can always be replaced by a ``short'' right
		U-turn, without loss of generality we suppose that~$C_i$ does not visit more than~$n^2$ cells to the right of 
		the left end-marker.
		\item $C_{k+1}$ starts from the state~$s_k$ with the head on the left end-marker and ends in~$q$, when
		the head leaves the segment, moving to the right of the cell~$m$.
		$C_{k+1}$ is called \emph{left-to-right traversal} of the frozen segment.
	\end{itemize}
		
	From~\cite[Theorem~3.5]{MP01}, there exists a set of positive 
	integers~$\{\ell_1,\ldots,\ell_r\}\subseteq\{1,\ldots,n\}$ satisfying~$\ell_1+\cdots+\ell_r\leq n$ such that for~$m\geq n$,
	if a frozen tape segment of length~$m$ can be (left-to-right or right-to-left) traversed from a state~$s$ 
	to a state~$s'$ then there is an index~$i\in\{1,\ldots,r\}$ such that, for each~$\mu>\frac{5n^2-m}{\ell_i}$,
	a frozen tape segment of length~$m+\mu\ell_i$ can be traversed (in the same direction)
	from state~$s$ to state~$s'$.
	This was proved by showing that for~$m>5n^2$ a traversal from~$s$ to~$s'$ of a segment of length~$m$ can always be ``pumped'' to
	obtain a traversal of a segment of length~$m'=m+\mu\ell_i$, for~$\mu>0$, and, furthermore, the segment can be ``unpumped''
	by taking~$\mu<0$, provided that the resulting length~$m'$ is greater than~$5n^2$.
	
	Let~$\ell$ be the least common multiple of~$\ell_1,\ldots,\ell_r$.
	If~$m>5n^2$, from the original computation path~$C$, by suitably pumping or unpumping the parts~$C_0$ and~$C_{k+1}$,
	and without changing~$C_i$, for~$i=1,\ldots,k$, for each~$m'=m+\mu\ell>5n^2$, with $\mu\in\IZ$,
	we can obtain a computation path that enters a frozen segment of length~$m'$ from the right in the state~$p$
	and exits the segment to the right in the state~$q$.
	\end{itemize}
\noindent
	By summarizing, from the previous analysis we conclude that for all~$m,m'>5n^2$, if~$m\equiv m'\pmod\ell$
	then~$\tau_m=\tau_{m'}$. Hence, the transition tables used in the simulation are at most~$5n^2+\ell$.
	Since, by definition, $\ell$ cannot exceed~$F(n)$, we obtain the number of different transitions tables that
	are used in the simulation is bounded by~$5n^2+1+F(n)$.
	
\medskip
	According with the construction outlined at the beginning of the proof, from the \fla~$M$ we can obtain a \ownfa~$A$ 
	that, when the head reaches the tape cell~$m+1$, has in the first component of its finite control the transition table~$\tau_m$,
	and in the second component the state in which the cell~$m+1$ is entered for
	the first time during the simulated computation. Hence the total number of states of~$A$ is bounded by~$n\cdot (5n^2+1+F(n))$.
	
	We observe that, at the beginning of the computation, the initial state is the pair containing the transition matrix~$\tau_0$
	and the initial state of~$M$. Hence, we do not need to consider other states with~$\tau_0$ as first component, unless~$\tau_0$
	occurs in the sequence~$\tau_1,\ldots,\tau_{5n^2+F(n)}$. This allows to reduce the upper bound to~$n\cdot (5n^2+F(n))+1$
	
	If the simulating automaton~$A$ is a \owdfa, then first component does not change, while the second component contains the set of
	states in which the cell~$m+1$ is entered for the first time during all possible computations
	of~$M$. This would give a~$2^n\cdot (5n^2+F(n))+1$ state upper bound.
	However, if the set in the second component is empty then the computation of~$M$ is rejecting, regardless what is the remaining
	part of the input and what has been written on the tape. Hence, in this case, the simulating \owdfa can
	enter a sink state.
	This allows to reduce the upper bound to~$(2^n-1)\cdot (5n^2+F(n))+2$.
\qed
\end{proof}

\subsection*{Optimality: The Language~${\cal L}_{n,\ell}$}

We now study the optimality of the state upper bounds presented
in \cref{th:upperBoundF-LA}.
To this aim, we introduce a family of languages~${\cal L}_{n,\ell}$,
that are defined with respect to integer parameters~$n,\ell>0$.

Each language in this family is composed by all strings
of length multiple of~$\ell$ belonging to the language~$L_{MF_n}$ which is accepted by the $n$-state 
\ownfa~$A_{MF_n}=(Q_n,\{a,b\},\delta_n,q_0,\{q_0\})$ depicted in \cref{fig:mf}, 
i.e., ${\cal L}_{n,\ell}=L_{MF_n}\cap (\{a,b\}^{\ell})^*$.

The automaton~$A_{MF_n}$ was proposed longtime ago by Meyer and Fischer as a witness of the exponential state
gap from \ownfas to \owdfas~\cite{MF71}. Indeed, it can be proved that the smallest \owdfa accepting it has exactly~$2^n$ states.
In the following we shall refer to some arguments given in the proof of such result presented in~\cite[Thm.~3.9.6]{Sha08}.

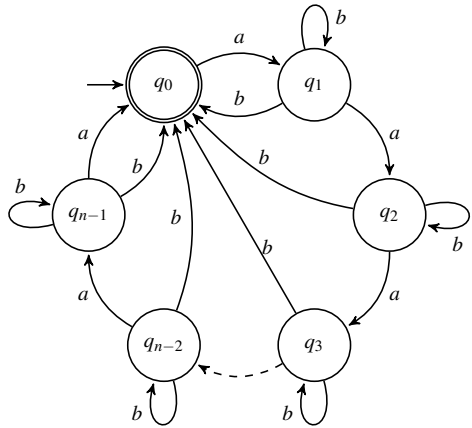
\begin{figure}[h]
	\centering
	\begin{tikzpicture}[font=\scriptsize]
		\def\stateangle{30}
		\def\statedistance{2cm}
		\path[inner sep=0]
			(0,0) node[state](5){$q_{n-1}$}
			
			++(60:\statedistance) node[state,accepting,initial,transition](0) {$q_0$}
			++(0:\statedistance) node[state](1){$q_1$}
			++(-60:\statedistance) node[state](2){$q_2$}

			(5)
			++(-60:\statedistance)  node[state](4){$q_{n-2}$}
			++(0:\statedistance)
		node[state](3){$q_3$};

		\path[transition]
			(1) edge[loop above] node[near end,right] {$b$} (1)
			(2) edge[loop right] node[near end,below] {$b$} (2)
			(3) edge[loop below] node[near end,left] {$b$} (3)
			(4) edge[loop below] node[near end,left] {$b$} (4)
			(5) edge[loop left] node[near end,above] {$b$} (5)

			(0) edge[bend left] node[above] {$a$} (1)
			(1) edge[bend left] node[right] {$a$} (2)
			(2) edge[bend left] node[right] {$a$} (3)
			(3) edge[bend left,dashed]  (4)
			(4) edge[bend left] node[left] {$a$} (5)
			(5) edge[bend left] node[left] {$a$} (0)

			(1) edge[bend left] node[above] {$b$} (0)
			(2.170) edge[out=170,in=-45] node[above] {$b$} (0.-45)
			(3) edge[] node[above,near start] {$b$} (0)
			(4.70) edge[out=70,in=-75] node[left] {$b$} (0.-75)
			(5) edge[bend right] node[left] {$b$} (0)
		;
	\end{tikzpicture}
	\caption{The \ownfa~$A_{MF_n}$ accepting the language of Meyer and Fischer.}
	\label{fig:mf}
\end{figure}

Let us start by presenting some simple state upper bounds for the recognition of~${\cal L}_{n,\ell}$ by
one-way finite automata.

\begin{theorem}
\label{th:Lnl-upperbound1}
	For every two integers~$n,\ell>0$, there exists a complete \owdfa accepting~${\cal L}_{n,\ell}$
	with~$(2^n-1)\cdot\ell+1$ states and a \ownfa with~$n\cdot\ell$ states.
\end{theorem}
\begin{proof}
	We apply the subset construction to convert the \ownfa~$A_{MF_n}$ into a \owdfa with~$2^n$ states and then, with the
	standard product construction, we intersect the resulting automaton with the trivial~$\ell$-state automaton accepting~$(\{a,b\}^{\ell})^*$.
	In this way we obtain a \owdfa with~$2^n\cdot\ell$ states for~${\cal L}_{n,\ell}$. However, all the states obtained from the
	sink state, corresponding to the empty set, are equivalent, so they can be replaced by a unique sink state.
	This allows to reduce the number of states to~$(2^n-1)\cdot\ell+1$.	
	
	In the case of \ownfas we apply the product construction to~$A_{MF_n}$ and the~$\ell$-state automaton accepting~$(\{a,b\}^{\ell})^*$,
	so obtaining a \ownfa  with~$n\cdot\ell$ states.
\qed
\end{proof}

We now study how to recognize~${\cal L}_{n,\ell}$ using two-way automata and \flas. In both cases we obtain sweeping machines.

\begin{theorem}
\label{th:Lnl-upperbound2}
Let~$\ell>0$ be an integer that factorizes~$\ell=p_1^{k_1}\cdots p_r^{k_r}$ as a product of prime powers
and~$o=r\bmod 2$. Then:
\begin{itemize}
	\item ${\cal L}_{n,\ell}$ is accepted by a sweeping \twnfa with~$n+p_1^{k_1}+\cdots+p_r^{k_r}+o$
	states, that uses nondeterministic transitions only in the first sweep.
	\item ${\cal L}_{n,\ell}$ is accepted by a sweeping \fla 
	with~$\max(n,p_1^{k_1}+\cdots+p_r^{k_r}+o)$ states that
	uses nondeterministic transitions only in the first sweep. 
	\item ${\cal L}_{n,\ell}$ is accepted by a sweeping \twdfa with~$2n+p_1^{k_1}+\cdots+p_r^{k_r}+o$
	states.
\end{itemize}
\end{theorem}
\begin{proof}
In the first sweep, the \twnfa for~${\cal L}_{n,\ell}$, using~$n$ states, simulates the \ownfa~$A_{MF_n}$
to check if the input belongs to~$L_{MF_n}$. 
Then, it makes one sweep for each~$i=1,\ldots,r$ (alternating a right-to-left sweep with a left-to-right sweep),
using~$p_i^{k_i}$ states in order to check whether~$p_i^{k_i}$ divides the input length. If the outcomes of all these
tests are positive, then the automaton accepts. When~$r$ is even, the last sweep ends with the head on the right end-marker.
Then, moving the head one position to the right, the automaton can reach the accepting configuration. 
However, when~$r$ is odd, the last sweep ends on the left end-marker. Hence, using an extra state, the head can traverse
the entire tape to finally reach the accepting configuration.

A \fla can implement the same strategy. However, to check if the tape length is a multiple of~$\ell$,
it can reuse the~$n$ states used in the first sweep, plus~$p_1^{k_1}+\cdots+p_r^{k_r}+o-n$ extra states 
when~$n < p_1^{k_1}+\cdots+p_r^{k_r}+o$. This is due to the fact that the value of the transition function
depends on the state and on the symbol in the tape cell and that, in the first sweep, all the input symbols
have been replaced by~$Z$. 

Finally, we can implement a \twdfa that recognizes~${\cal L}_{n,\ell}$ by firstly making~$r$ sweeps to
check whether~$p_i^{k_i}$ divides the input length, $i=1,\ldots,r$. If so, then the automaton,
after moving the head from the left to the right end-marker in case of~$r$ 
even, makes a further sweep from right to left, to simulate a \owdfa accepting the reversal of~$L_{MF_n}$, which
can be accepted using~$2n$ states~\cite{PPS22}. If the simulated automaton accepts, then the machine can make a further sweep,
by using a unique state to move the head from the left end-marker to the right one, and then accept.
The total number of states is~$2n+p_1^{k_1}+\cdots+p_r^{k_r}+2-o$. This number can be slightly reduced as follows:
in the first sweep (which is from left to right) the automaton checks the divisibility of the input length by~$p_1^{k_1}$;
in the second sweep (from right to left) the automaton checks the membership to~$L_{MF_n}$; in the remaining~$r-1$
sweeps (alternating left-to-right with right-to-left sweeps), it checks the divisibility for~$p_i^{k_i}$, $i=2,\ldots,r$. 
So, the total number of sweeps for these checks is~$r+1$.
This means that, when~$r$ is even, the last sweep ends on the right end-marker and the machine can immediately move
to the accepting configuration. Otherwise the head needs to cross the input from left to right, using an extra state.
 \qed
\end{proof}

As a consequence of \cref{th:Lnl-upperbound2}, in the case of \flas we immediately obtain:

\begin{corollary}
\label{cor:Lnl-upperbound3}
	For each~$n>0$ the language~${\cal L}_{n,F(n)}$ is accepted by a \fla with at most~$n+1$ states.
\end{corollary}
\begin{proof}
	If~$F(n)=p_1^{k_1}\cdots p_r^{k_r}$ then $p_1^{k_1}+\cdots+p_r^{k_r}\leq n\leq F(n)$.
	Hence, the statement follows from \cref{th:Lnl-upperbound2}.
\qed
\end{proof}

We are now going to prove lower bounds for the recognition of~${\cal L}_{n,\ell}$, in the case~$n$ and~$\ell$ are
relatively primes.

Let us start by considering the recognition by \owdfas. 

\begin{theorem}
\label{th:Lnl-lowerbound}
	Given two integers~$n,\ell>0$ with~$\gcd(n,\ell)=1$, each \owdfa accepting~${\cal L}_{n,\ell}$ must have at least~$(2^n-1)\cdot\ell+1$ states.
\end{theorem}
\begin{proof}
	Let~$Q_n=\{q_0,q_1,\ldots,q_{n-1}\}$ be the set of states of~$A_{MF_n}$ (see \cref{fig:mf}).
	First, we briefly recall some arguments from the proof presented in~\cite[Thm.~3.9.6]{Sha08}.
	For each subset~$S$ of~$Q_n$,
	we define a string~$w_S$ having the property that~$\delta_n(q_0,w_S)=S$.
	Furthermore,
	it is proved that all the strings so defined are pairwise distinguishable, so obtaining the state lower bound~$2^n$ 
	for each \owdfa equivalent to~$A_{MF_n}$.
	In particular, the string~$w_S$ is defined as follows:
	\begin{equation}
	w_S=\left\{
			\begin{array}{ll}
				b & \mbox{if~$S=\emptyset$};\\
				a^i&\mbox{if $S=\{q_i\}$};\\		
				a^{e_k-e_{k-1}}ba^{e_{k-1}-e_{k-2}}b\cdots a^{e_2-e_1}ba^{e_1},
				&\mbox{otherwise};\\
			\end{array}	
		\right.
		\label{eq:mf-shallit}
	\end{equation}
	where in the second case $S=\{q_i\}$, $0\leq i<n$, while in the third case~$S=\{q_{e_1},q_{e_2},\ldots,q_{e_k}\}$,
	$1<k\leq n$, and~$0\leq e_1<e_2<\cdots<e_k<n$.
	
	To obtain the claimed state lower bound in the case of the language~${\cal L}_{n,\ell}$,
	for each nonempty subset~$S$ of~$Q_n$ and each integer~$j$, with~$0\leq j<\ell$, 
	we define a string~$w_{S,j}$ which is obtained by suitably padding the string~$w_S$ in such a way
	that the set of states reachable from the initial state by reading~$w_{S,j}$ remains~$S$ and the length
	of~$w_{S,j}$, divided by~$\ell$, gives~$j$ as reminder.
	Then we shall prove that all the so obtained strings are pairwise distinguishable.
	Unlike~(\ref{eq:mf-shallit}), when defining~$w_{S,j}$ we do not consider the case~$S=\emptyset$.
	
	In the following, let us denote by~$f:\IN\times\IN\rightarrow\IN$ a function satisfying~$f(i,j)\bmod n=i$ and~$f(i,j)\bmod\ell=j$,
	for~$i,j\in\IN$. Since~$\gcd(n,\ell)=1$, by the Chinese Reminder Theorem, such a function always exists.
	
	For each non-empty subset~$S$ of~$Q_n$ and each integer~$j$, with~$0\leq j<\ell$, the string~$w_{S,j}$ is defined as:
	\begin{equation}
	w_{S,j}=\left\{
			\begin{array}{ll}
				a^{f(i,j)}&\mbox{if $S=\{q_i\}$};\\		
				a^{e_k-e_{k-1}}ba^{e_{k-1}-e_{k-2}}b\cdots a^{e_2-e_1}b^{H\ell-k-e_{k}+2+j}a^{e_1},
				&\mbox{otherwise};\\
			\end{array}	
		\right.
		\label{eq:mf-modified}
	\end{equation}
	where in the first case $S=\{q_i\}$, $0\leq i<n$, while in the second case~$S=\{q_{e_1},q_{e_2},\ldots,q_{e_k}\}$,
	$1<k\leq n$, $0\leq e_1<e_2<\cdots<e_k<n$, and~$H\geq 1$ is a fixed integer such that~$H\ell > 2n$ 
	(this last condition is useful to have~$H\ell-k-e_{k}+2+j>0$, in such a way that the last block of~$b$'s is
	always well defined and not empty).
	
	We claim and prove the following facts:
		
	\begin{enumerate}
		\item $|w_{S,j}|\bmod\ell = j$.\\
		If~$S=\{q_i\}$, then by definition $|w_{S,j}|\bmod\ell=f(i,j)\bmod\ell =j$.
		Otherwise, according to the second case in~(\ref{eq:mf-modified}), $S=\{q_{e_1},q_{e_2},\ldots,q_{e_k}\}$
		and~$|w_{S,j}|= e_k-e_{k-1} + 1 + e_{k-1}-e_{k-2} + 1 + \cdots + e_2-e_1 + H\ell-k-e_k+2+j + e_1$,
		which is equal to~$H\ell + j$.
				
		\item $\delta_n(q_0,w_{S,j})=S$.\\
		In the automaton~$A_{MF_n}$, all the transitions on the letter~$a$ are deterministic. Furthermore,
		by reading the string~$a^x$, $x>0$, from the state~$q_0$, the only reachable state is~$q_{x\bmod n}$.
		Hence, for the first case~$S=\{q_i\}$ in~(\ref{eq:mf-modified}) we have~$\delta_n(q_0,w_{S,j})=
		\{q_{f(i,j)\bmod n}\} = \{q_i\}$.
		
		For the second case,  we already mentioned that~$\delta_n(q_0,w_S)=S$.
		Furthermore~$w_{S,j}$ is obtained from~$w_S$ by replacing the rightmost~$b$
		by a block of more than one~$b$. From the transition diagram of~$A_{MF_n}$ we  observe that from each
		state~$q_i$, with~$i>0$, reading a~$b$ the automaton can either remain in~$q_i$ or move to~$q_0$.
		Furthermore, from~$q_0$ there are no transitions on the letter~$b$. This allows to conclude that
		the behavior does not change when one replaces an occurrence of~$b$ in a string with a sequence of more than one~$b$.
		Hence, $\delta_n(q_0,w_{S,j})=\delta_n(q_0,w_S)=S$.
		
		\item\emph{For $i=0,\ldots,n-1$ and $x\geq 0$, $\delta_n(q_i,a^x)=\{q_{i'}\}$ where~$i'=0$
		if and only if~$x \bmod n = n - i$. Hence $a^x$ is accepted by some computation path
		starting from~$q_i$ if and only if~$x \bmod n = n - i$.}\\
		It is enough to observe that all the transitions on the letter~$a$ are deterministic
		and form a loop visiting all the states. More precisely~$i'=(i+x)\bmod n$. Hence, $i'=0$
		if and only if~$x \bmod n = n - i$.
	\end{enumerate}
	We now prove that all the strings~$w_{S,j}$ are pairwise distinguishable.
	To this aim, let us consider two such strings~$w_{S,j}$ and~$w_{T,h}$, with $(S,j)\neq(T,h)$.
	We inspect the following two cases:
	\begin{itemize}
		\item $S\neq T$.
		Without loss of generality, let us consider a state~$q_s\in S\setminus T$. We 
		take~$z=a^{f(n-s,\ell-j)}$. By the previous claims, we obtain that~$w_{S,j}\cdot z\in L_{MF_n}$,
		while~$w_{T,h}\cdot z\notin L_{MF_n}$. Furthermore,~$|w_{S,j}\cdot z|\bmod\ell=(j + \ell-j)\bmod\ell = 0$.
		Hence~$w_{S,j}\cdot z\in(\{a,b\}^{\ell})^*$. This allows to conclude that~$w_{S,j}\cdot z\in{\cal L}_{n,\ell}$,
		while~$w_{T,h}\cdot z\notin{\cal L}_{n,\ell}$.		
		\item $j\neq h$.
		We choose a state~$q_s\in S$ and, again, the string~$z=a^{f(n-s,\ell-j)}$.
		Exactly as in the previous case we obtain~$w_{S,j}\cdot z\in{\cal L}_{n,\ell}$.
		Furthermore, being~$j\neq h$ and~$0\leq j,h<\ell$, we get that~$|w_{T,h}\cdot z|\bmod\ell=(h + \ell-j)\bmod\ell \neq 0$.
		Hence~$w_{T,h}\cdot z\notin(\{a,b\}^{\ell})^*$, thus implying~$w_{T,h}\cdot z\notin{\cal L}_{n,\ell}$.	
	\end{itemize}
	By summarizing, we have proved that all the above defined~$(2^n-1)\cdot\ell$ strings~$w_{S,j}$ are pairwise distinguishable.
	We also observe that each string starting with the letter~$b$ is not accepted by the automaton~$A_{MF_n}$.\footnote{%
	We point out that two strings that in~$A_{MF_n}$ lead to the emptyset are not distinguishable. This is the
	reason why we did not considered strings of the form~$w_{\emptyset,j}$ in~(\ref{eq:mf-modified}).
	}
	This implies that the string~$b$ and each string~$w_{S,j}$ are distinguishable.
	Hence, we are able to conclude that each \owdfa accepting~${\cal L}_{n,\ell}$ has at least~$(2^n-1)\cdot\ell+1$ states.
\qed	
\end{proof}

Concerning \ownfas, we prove the following:

\begin{theorem}
\label{th:Lnl-lowerboundNFA}
	Given two integers~$n,\ell>0$ with~$\gcd(n,\ell)=1$, each \ownfa accepting~${\cal L}_{n,\ell}$ 
	must have at least~$n\cdot\ell$ states.
\end{theorem}
\begin{proof}
	The proof can be easily given by observing that~$X=\{(a^i,a^{n\cdot\ell-i}) \;\mid\; i=0,\ldots,n\cdot\ell-1 \}$ 
	is a \emph{fooling set} for~${\cal L}_{n,\ell}$~\cite{Bir92}. Hence, the number of states of each \ownfa\
	for~${\cal L}_{n,\ell}$ cannot be lower than the cardinality of~$X$.
\qed
\end{proof}

As a consequence of \cref{th:Lnl-lowerbound,th:Lnl-lowerboundNFA} we obtain:

\begin{theorem}
\label{th:LnF(n)-lowerbound}
	For each prime~$n>4$, every \owdfa and every \ownfa accepting~${\cal L}_{n,F(n)}$ needs~$(2^n-1)\cdot F(n)+1$ 
	and~$n\cdot F(n)$ states, respectively.
\end{theorem}
\begin{proof}
	First, we prove that~$\gcd(n,F(n))=1$ for each prime~$n>4$.
	To this aim, we observe that by definition $F(n)\geq 2\cdot(n-2)$ for each prime~$n$. 
	Furthermore, if~$n>4$ then~$2\cdot(n-2)>n$.
	Hence~$F(n)>n$ for each prime~$n>4$.
	Suppose that~$\gcd(n,F(n))\neq 1$. Then~$n$, being prime and less than~$F(n)$, should divide~$F(n)$.
	By definition of~$F(n)$, this
	would imply~$F(n)=n$; a contradiction. 
	This allows us to conclude that~$\gcd(n,F(n))=1$, for each prime~$n>4$.

	Using \cref{th:Lnl-lowerbound,th:Lnl-lowerboundNFA},
	we get that,
	for all such~$n$'s,
	a \owdfa needs at least~$(2^n-1)\cdot F(n)+1$ states to accept~${\cal L}_{n,F(n)}$,
	while an equivalent \ownfa needs at least~$n\cdot\ell$ states.
\qed
\end{proof}

As a consequence of \cref{th:LnF(n)-lowerbound}, for infinitely many~$n$, the
\owdfa and \ownfa for the language~${\cal L}_{n,F(n)}$ described in \cref{th:Lnl-upperbound1}
are minimal.

By combining the results in \cref{cor:Lnl-upperbound3,th:LnF(n)-lowerbound},
we obtain that the costs of the simulations of \flas by \ownfas and \owdfas presented in
\cref{th:upperBoundF-LA} are asymptotically optimal:
\begin{corollary}
	For infinitely many integers~$n$ there exists a language which is accepted by a \fla\
	with at most~$n+1$ states and such that all equivalent \owdfas and \ownfas require at least~$(2^n-1)\cdot F(n)+1$
	and~$n\cdot F(n)$ states, respectively. 
\end{corollary}

\section{Deterministic Forgetting $1$-Limited Automata vs.\ One-Way Automata}
\label{sec:DETtoOneWay}

In \cref{sec:toOneWay} we studied the size costs of the conversions from \flas to one-way finite automata.
We now restrict our attention to the simulation of deterministic machines.
By adapting to this case the arguments used to prove \cref{th:upperBoundF-LA},
we obtain a superpolynomial state bound for the conversion into \owdfas,
which is not so far from the bound obtained starting from nondeterminstic machines:

\begin{theorem}
\label{th:upperBoundDF-LA}
  Let~$M$ be an~$n$-state \dfla.
  Then $M$ can be simulated by a \owdfa with at most~$n\cdot(n+F(n))+2$ states.
\end{theorem}
\begin{proof}
	We can apply the construction given in the proof of \cref{th:upperBoundF-LA} to build,
	from the given \dfla~$M$,
	a one-way finite automaton that, when the head reaches the tape cell~$m+1$, has in its finite control
	the transition table~$\tau_m$ associated with the tape segment of length~$m$ and the state in which the
	cell is reached for the first time. Since the transitions of~$M$ are deterministic,
	each tape cell is reached for the first time by at most one computation and the resulting
	automaton is a (possible partial) \owdfa, with no more than~$n\cdot (5n^2+F(n))+1$ states.
	However, in this case the number of transition tables can be reduced, so decreasing the upper bound.
	In particular, due to determinism and the unary content in the frozen part, 
	we can observe that left and right U-turns cannot visit more than~$n$ tape
	cells. %
	Furthermore, after visiting more than~$n$ tape cells, a traversal is repeating a loop.
	This allows to show that the sequence of transition matrices starts to be periodic after
	the matrix~$\tau_n$,
	i.e, for~$m,m'>n$, if~$m\equiv m'\pmod{ F(n)}$ then~$\tau_m=\tau_{m'}$. 
	Hence, the number of different transition tables used during the simulation is at most~$n+1+F(n)$,
	and the number of states of the simulating  (possibly partial) \owdfa is bounded by~$n\cdot(n+F(n))+1$.
	By adding one more state we finally obtain a complete \owdfa.
\qed
\end{proof}

\subsection*{Optimality: The Language~${\cal J}_{n,\ell}$}

We now present a family of languages for which we prove a size gap very close
to the upper bound in \cref{th:upperBoundDF-LA}.
Given two integers~$n,\ell>0$, let us consider:
\[
	{\cal J}_{n,\ell} = \{ w\in\{a,b\}^* \;\mid\; |w|_a\bmod n = 0 \mbox{ and } |w| \bmod \ell = 0 \}\,.
\]
First of all, we observe that it is not difficult to recognize ${\cal J}_{n,\ell}$ using a \owdfa with~$n\cdot\ell$ states that
counts the number of~$a$'s using one counter modulo~$n$,
and the input length using one counter modulo~$\ell$.
This number of states cannot be reduced,
even allowing nondeterministic transitions:

\begin{theorem}
\label{th:lbJ}
	Each \ownfa accepting ${\cal J}_{n,\ell}$ has at least~$n\cdot\ell$ states.
\end{theorem}
\begin{proof}
	Let~$H>\ell+n$ be a multiple of~$\ell$. For~$i=1,\ldots,\ell$, $j=0,\ldots,n-1$,
	consider~$x_{ij}=a^jb^{H+i-j}$ and~$y_{ij}=b^{H-i-n+j}a^{n-j}$.
	We are going to prove that the set
	\[
		X = \{ (x_{ij},y_{ij} ) \;\mid\; 1\leq i\leq\ell, 0\leq j < n \}
	\]
	is an \emph{extended fooling set} for~${\cal J}_{n,\ell}$.
	To this aim, let us consider~$i,i'=1,\ldots,\ell$, $j,j'=0,\ldots,n-1$.
	We observe that the string~$x_{ij}y_{ij}$ contains~$n$ $a$'s and has 
	length~$j+H+i-j+H-i-n+j+n-j=2H$ and hence it belongs to~${\cal J}_{n,\ell}$.
	For~$i,i'=1,\ldots,\ell$, if~$i\neq i'$ then the string~$x_{ij}y_{i'j}\notin{\cal J}_{n,\ell}$
	because it has length~$2H+i-i'$ which cannot be a multiple of~$\ell$.
	On the other hand, if~$j<j'$, the string~$x_{ij}y_{i'j'}$ contains~$j+n-j'<n$ many~$a$'s, so it
	cannot belong to~${\cal J}_{n,\ell}$,
\qed
\end{proof}

Concerning the recognition of~${\cal J}_{n,\ell}$ by \flas we prove the following:

\begin{theorem}
\label{th:Jnl-LF}
	Let~$\ell>0$ be an integer that factorizes~$\ell=p_1^{k_1}\cdots p_s^{k_r}$ as product of prime powers, $o=r\bmod 2$,
	and~$n>0$. Then~${\cal J}_{n,\ell}$ is accepted by a sweeping \twdfa with~$n+p_1^{k_1}+ \cdots + p_r^{k_r}+o$ states
	and by a sweeping \dfla with~$\max(n,p_1^{k_1}+ \cdots + p_r^{k_r}+o)$ states.	
\end{theorem}
\begin{proof}
	A \twdfa can make a first sweep of the input, using~$n$ states, to check if the number of~$a$'s in the input is a multiple of~$n$.
	Then, in further~$r$ sweeps, alternating right-to-left with left-to-right sweeps, it can check the divisibility of the
	input length by~$p_i^{k_i}$, $i=1,\dots,r$. If~$r$ is odd this process ends with the head on the left end-marker.
	Hence, in this case,
	when all tests are positive,
	a further sweep (made by using a unique state) is used to move the head from the left
	to the right end-marker and then reach the accepting configuration.
	
	We can implement a \dfla that uses the same strategy.
	However, after the first sweep, all input symbols are  replaced
	by~$Z$. Hence, as in the proof of \cref{th:Lnl-upperbound2}, the machine can reuse the~$n$ states
	of the first sweep. So, the total number of states reduces to~$\max(n,p_1^{k_1}+  \cdots + p_r^{k_r}+o)$.
\qed
\end{proof}

As a consequence of \cref{th:Jnl-LF}, we obtain:

\begin{corollary}
\label{cor:Jnl-LF}
	For each integer~$n>0$ the language~${\cal J}_{n,F(n)}$ is accepted by a \dfla with at most~$n+1$ states.
\end{corollary}

By combining the upper bound in \cref{cor:Jnl-LF} with the lower bound in \cref{th:lbJ},
we obtain that the superpolynomial cost of the simulation of \dflas by \owdfas given in \cref{th:upperBoundDF-LA} is asymptotically optimal and
it cannot be reduced even if the resulting automaton is nondeterministic:
\begin{corollary}
	For each integer~$n>0$ there exists a language accepted by a \dfla\
	with at most~$n+1$ states and such that all equivalent \owdfas and \ownfas require at least~$n\cdot F(n)$
	states. 
\end{corollary}

\section{Forgetting $1$-Limited vs.\ Two-Way Automata}
\label{sec:FvsTW}

Up to now, we have studied the size costs of the transformations of \flas and \dflas into one-way automata.
We proved that they cannot be significantly reduced, by providing suitable witness languages.
However, we can notice that such languages are accepted by two-way automata whose sizes are not so far from
the sizes of \flas and \dflas we gave.
So we now analyze the size relationships between forgetting and two-way automata.
On the one hand,
we show that forgetting input symbols can dramatically reduce the descriptional power.
Indeed, we provide a family of languages for which \flas are exponentially larger than \twdfas.
On the other hand, we guess that also in the opposite direction at least a superpolynomial gap can be possible.
To this aim we present a language accepted by a \dfla of size~$O(n)$ and we conjuncture that each \twnfa accepting
it requires more than~$F(n)$ states.

\subsection*{From Two-way to Forgetting $1$-Limited Automata}

For each integer~$n>0$, let us consider the following language
\[
	\mathcal E_n=\{w\in\{a,b\}^* \;\mid\; \exists x\in\{a,b\}^n, \exists y,z\in\{a,b\}^*: w=x\cdot y = z\cdot x^R\}\text,
\]
i.e., the set of strings in which the prefix of length~$n$ is equal to the reversal of the suffix.
As we shall see,
it is possible to obtain a \twdfa with~$\bigoof{n}$ states accepting it.
Furthermore, each equivalent \fla requires~$2^n$ states.

To achieve this result, first we give a lower bound technique for the number of states of \flas, 
which is inspired by the \emph{fooling set technique} for \ownfas~\cite{Bir92}.

\begin{lemma}
\label{lemma:fooling}
	Let~$L\subseteq\Sigma^*$ be a language and~$X=\{(x_i,y_i)\;\mid\; i=1,\ldots,n\}$ be a set of words
	such that the following hold:
	\begin{itemize}
		\item $|x_1|=|x_2|=\cdots=|x_n|$,
		\item $x_iy_i\in L$, for $i=1,\ldots,n$,
		\item $x_iy_j\notin L$ or~$x_jy_i\notin L$, for $i,j=1,\ldots,n$ with~$i\neq j$.
	\end{itemize}
	Then each \flas accepting~$L$ has at least~$n$ states.
\end{lemma} 
\begin{proof}
	Let~$M$ be a \flas accepting~$L$. Let~$C_i$ be an accepting computation of~$M$ on input~$x_iy_i$, $i=1,\ldots,n$.
	We divide~$C_i$ into two parts~$C'_i$ and~$C''_i$, where~$C'_i$ is the part of~$C_i$ that starts from the initial configuration
	and ends when the head reaches for the first  time the first cell to the right of~$x_i$, namely the cell containing the first symbol of~$y_i$,
	while~$C''_i$ is the remaining part of~$C_i$.
	Let~$q_i$ be the state reached at the end of~$C'_i$, namely the state from which~$C''_i$ starts.

	If~$q_i=q_j$, for some~$1\leq i,j\leq n$, then the computation obtained concatenating~$C'_i$ and~$C''_j$ accepts the input~$x_iy_j$.
	Indeed, at the end of~$C'_i$ and of~$C'_j$, the content of the tape to the
	left of the head is replaced by the same string~$Z^{|x_i|}=Z^{|x_j|}$. 
	So~$M$, after inspecting~$x_i$, can perform exactly the same moves as
	on input~$x_jy_j$ after inspecting~$x_j$ and hence it can accept~$x_iy_j$.
	In a similar way, concatenating~$C'_j$ and~$C''_i$ we obtain an accepting
	computation on~$x_jy_i$. If~$i\neq j$, then this is a contradiction.
	
	This allows to conclude that~$n$ different states are necessary for~$M$.
\qed
\end{proof}

We are now able to prove the claimed separation.

\begin{theorem}
\label{th:2DFA->F}
	The language~$\mathcal E_n$ is accepted by a \twdfa with~$\bigoof{n}$ states, while each
	\fla accepting it has at least~$2^n$ states.
\end{theorem}
\begin{proof}
	We can build a \twdfa that on input~$w\in\Sigma^*$ tests the equality between the symbols in positions~$i$ and~$|w|-i$ of~$w$,
	for~$i=1,\ldots,n$. If one of the tests fails, then the automaton stops and rejects, otherwise it finally accepts.
	For each~$i$, the test starts with the head on the left end-marker and the value of~$i$ in the finite control.
	Then, the head is moved to the right, while decrementing~$i$, to locate the~$i$th input cell and remember its content in the finite control.
	At this point, the head is moved back to the left end-marked, while counting input cells to restore the value of~$i$.
	The input is completely crossed from left to right, by keeping this value in the control. When the right end-marker is reached,
	a similar procedure is applied to locate the symbol in position~$|w|-i$, which is then compared with that in position~$i$, previously stored in
	the control. If the two symbols are equal, then the head is moved again to the right end-marker, while restoring~$i$.
	If~$i=n$, then the machine moves in the accepting configuration, otherwise the value of~$i$ is incremented and the head is moved to the left 
	end-marker to prepare the next test.
	From the above description we can conclude that~$\bigoof{n}$ states are enough for a \twdfa to accept~$\mathcal E_n$.
	
	For the lower bound, we observe that the set~$X=\{(x,x^R)\;\mid\; x\in\{a,b\}^n\}$,
	whose cardinality is~$2^n$, satisfies the requirements of \cref{lemma:fooling}.
\qed
\end{proof}

\subsection*{From Forgetting $1$-limited to Two-way Automata}

We wonder if there is some language showing an exponential, or at least superpolynomial, 
size gap from \flas to two-way automata.
Here we propose, as a possible candidate, the following language, where~$n,\ell>0$ are integers:
\[
	{\cal H}_{n,\ell} = \{\, ub^nv \;\mid\; u\in(a+b)^*a,\,v\in(a+b)^*,\, |u|_a\bmod n = 0, \mbox{ and~} |u| \bmod \ell = 0 \}\text.
\]
We prove that~${\cal H}_{n,F(n)}$ can be recognized by a \dfla with a number of
states linear in~$n$.

\begin{theorem}
	For each integer~$n>1$ the language~${\cal H}_{n,F(n)}$ is accepted by a \dfla with~$\bigoof{n}$ states.
\end{theorem}
\begin{proof}
	A \dfla~$M$ can start to inspect the input from left to right, while counting modulo~$n$
	the~$a$'s.
	In this way it can discover each prefix~$u$ that ends with an~$a$ and such that~$|u|_a\bmod n = 0$.
	When such a prefix is located, $M$ verifies whether~$|u|$ is a multiple of~$F(n)$
	and it is followed by~$b^n$. We will discuss how to do that below.
	If the result of the verification is positive, then~$M$ moves to the accepting configuration,
	otherwise it continues the same process.
	
	Now we explain how the verification can be performed. Suppose~$F(n)=p_1^{k_1}\cdots p_r^{k_r}$, where~$p_1^{k_1},\ldots,p_r^{k_r}$
	are prime powers. First, we point out that when the verification starts,
	exactly the first~$|u|$ tape cells have been rewritten. Hence, the rough idea is to alternate right-to-left
	and left-to-right sweeps on such a portion of the tape, to check the divisibility of~$|u|$ by each~$p_i^{k_i}$,
	$i=1,\ldots,r$.
	A right-to-left sweep stops when the head reaches the left end-marker.
	On the other hand,
	a left-to-right sweep can end only when the head reaches the first cell to the right of the frozen
	segment. This forces the replacement of the symbol in it with the symbol~$Z$, so increasing the length
	of the frozen segment by~$1$. 	
	In the next sweeps, the machine has to take into account how much the frozen segment increased.
	For instance, after checking divisibility by~$p_1^{k_1}$ and by~$p_2^{k_2}$, in the next sweep the
	machine should verify that the length of the frozen segment, modulo~$p_3^{k_3}$, is~$1$. 
	Because the machine has to check~$r$ divisors and right-to-left sweeps alternate with left-to-right
	sweeps, when all~$r$ sweeps are done, exactly~$\lfloor r/2\rfloor$ extra cells to the right of
	the original input prefix~$u$ are frozen. Since $n>r/2$, if the original symbol in all those cells
	was~$b$,
	to complete the verification phase the machine has to check whether
	the next~$n-\lfloor r/2\rfloor$ not yet visited cells contain~$b$.
	However,
	the verification fails
	if a cell containing an~$a$ or the right end-marker is reached during some point of the verification phase.
	This can happen either while checking the length of the frozen segment
	or while checking the last~$n-\lfloor r/2\rfloor$ cells.
	If the right end-marker is reached, then the machine rejects. Otherwise
	it returns to the main procedure, i.e., resumes the counting of the~$a$'s.
	
	The machine uses a counter modulo~$n$ for the~$a$'s. In the verification phase this counter keeps the value~$0$.
	The device first has to count the length of the frozen part modulo~$p_i^{k_i}$, iteratively for~$i=1,\ldots,r$, and to
	verify that the inspected prefix is followed by~$b^n$, using again a counter.
	Since~$p_1^{k_1}+ \cdots + p_r^{k_r}\leq n$, by summing up we conclude that the total number of 
	states is~$\bigoof{n}$.
\qed
\end{proof}

By using a modification of the argument in the proof of \cref{th:lbJ}, we can show
that each \ownfa accepting~${\cal H}_{n,F(n)}$ cannot have less than~$n\cdot F(n)$ states.\footnote{%
It is enough to consider the set~$X'=\{ (x_{ij},y_{ij}b^n ) \;\mid\; 1\leq i\leq\ell, 0\leq j < n \}$,
instead of~$X$.}
We guess that such a number cannot be substantially reduced even having the possibility of moving the head in both directions.
In fact, a two-way automaton using~$\bigoof{n}$ states can easily locate on the input tape a ``candidate''
prefix~$u$. However, it cannot remember in which position of the tape~$u$ ends, in order to check~$|u|$ in several sweeps 
of~$u$. So we do not see how the machine could verify whether~$|u|$ is a multiple of~$F(n)$ using less than~$F(n)$ states.

\section{Conclusion}

We compared the size of forgetting $1$-limited automata with that of finite automata,
proving exponential and superpolynomial gaps.
We did not discuss the size relationships with \las.
However, since \twdfas are \dlas that never write, as a corollary of \cref{th:2DFA->F} we  get an
exponential size gap from \dlas to \flas. Indeed, the fact of having a unique symbol
to rewrite the tape content dramatically reduces the descriptional power.

We point out that this reduction happens also in the case of \flas accepting languages defined over a one-letter
alphabet, namely unary languages.
To this aim, for each integer~$n>0$, let us consider the language~$(a^{2^n})^*$.
This language can be accepted with a \dla having~$\bigoof{n}$ states and a work
alphabet of cardinality~$\bigoof{n}$, and with a \dla having~$\bigoof{n^3}$ states and 
a work alphabet of size not dependent on~$n$~\cite{PP19,PP23}.
However, each \twnfa accepting it requires at least~$2^n$ states~\cite{PP19}.
Considering the cost of the conversion of \flas into \ownfas (\cref{th:upperBoundF-LA}), we can conclude that
such a language cannot be accepted by any \fla having a number of states polynomial in~$n$.

\bibliographystyle{eptcs}
\bibliography{biblio}

\begin{thebibliography}{10}
\providecommand{\bibitemdeclare}[2]{}
\providecommand{\surnamestart}{}
\providecommand{\surnameend}{}
\providecommand{\urlprefix}{Available at }
\providecommand{\url}[1]{\texttt{#1}}
\providecommand{\href}[2]{\texttt{#2}}
\providecommand{\urlalt}[2]{\href{#1}{#2}}
\providecommand{\doi}[1]{doi:\urlalt{https://doi.org/#1}{#1}}
\providecommand{\eprint}[1]{arXiv:\urlalt{https://arxiv.org/abs/#1}{#1}}
\providecommand{\bibinfo}[2]{#2}

\bibitemdeclare{article}{Bir92}
\bibitem{Bir92}
\bibinfo{author}{Jean{-}Camille \surnamestart Birget\surnameend}
  (\bibinfo{year}{1992}): \emph{\bibinfo{title}{Intersection and Union of
  Regular Languages and State Complexity}}.
\newblock {\slshape \bibinfo{journal}{Inf. Process. Lett.}}
  \bibinfo{volume}{43}(\bibinfo{number}{4}), pp. \bibinfo{pages}{185--190},
  \doi{10.1016/0020-0190(92)90198-5}.

\bibitemdeclare{article}{Ch86}
\bibitem{Ch86}
\bibinfo{author}{Marek \surnamestart Chrobak\surnameend}
  (\bibinfo{year}{1986}): \emph{\bibinfo{title}{Finite Automata and Unary
  Languages}}.
\newblock {\slshape \bibinfo{journal}{Theor. Comput. Sci.}}
  \bibinfo{volume}{47}(\bibinfo{number}{3}), pp. \bibinfo{pages}{149--158}.
\newblock \urlprefix\url{http://dx.doi.org/10.1016/0304-3975(86)90142-8}.
\newblock \bibinfo{note}{Errata: \cite{Ch03}}.

\bibitemdeclare{article}{Ch03}
\bibitem{Ch03}
\bibinfo{author}{Marek \surnamestart Chrobak\surnameend}
  (\bibinfo{year}{2003}): \emph{\bibinfo{title}{Errata to: Finite automata and
  unary languages: [{Theoret. Comput. Sci.} 47 (1986) 149-158]}}.
\newblock {\slshape \bibinfo{journal}{Theor. Comput. Sci.}}
  \bibinfo{volume}{302}(\bibinfo{number}{1-3}), pp. \bibinfo{pages}{497 --
  498}, \doi{10.1016/S0304-3975(03)00136-1}.

\bibitemdeclare{article}{Ge07}
\bibitem{Ge07}
\bibinfo{author}{Viliam \surnamestart Geffert\surnameend}
  (\bibinfo{year}{2007}): \emph{\bibinfo{title}{Magic numbers in the state
  hierarchy of finite automata}}.
\newblock {\slshape \bibinfo{journal}{Inf. Comput.}}
  \bibinfo{volume}{205}(\bibinfo{number}{11}), pp. \bibinfo{pages}{1652--1670}.
\newblock \urlprefix\url{http://dx.doi.org/10.1016/j.ic.2007.07.001}.

\bibitemdeclare{article}{GP12}
\bibitem{GP12}
\bibinfo{author}{Viliam \surnamestart Geffert\surnameend} \&
  \bibinfo{author}{Giovanni \surnamestart Pighizzini\surnameend}
  (\bibinfo{year}{2012}): \emph{\bibinfo{title}{Pairs of Complementary Unary
  Languages with ``Balanced'' Nondeterministic Automata}}.
\newblock {\slshape \bibinfo{journal}{Algorithmica}}
  \bibinfo{volume}{63}(\bibinfo{number}{3}), pp. \bibinfo{pages}{571--587},
  \doi{10.1007/s00453-010-9479-9}.

\bibitemdeclare{article}{Hi67}
\bibitem{Hi67}
\bibinfo{author}{Thomas~N. \surnamestart Hibbard\surnameend}
  (\bibinfo{year}{1967}): \emph{\bibinfo{title}{A Generalization of
  Context-Free Determinism}}.
\newblock {\slshape \bibinfo{journal}{Inf. Control.}}
  \bibinfo{volume}{11}(\bibinfo{number}{1/2}), pp. \bibinfo{pages}{196--238},
  \doi{10.1016/S0019-9958(67)90513-X}.

\bibitemdeclare{book}{HU79}
\bibitem{HU79}
\bibinfo{author}{John~E. \surnamestart Hopcroft\surnameend} \&
  \bibinfo{author}{Jeffrey~D. \surnamestart Ullman\surnameend}
  (\bibinfo{year}{1979}): \emph{\bibinfo{title}{Introduction to Automata
  Theory, Languages and Computation}}.
\newblock \bibinfo{publisher}{Addison-Wesley}.

\bibitemdeclare{inproceedings}{JMP93}
\bibitem{JMP93}
\bibinfo{author}{Petr \surnamestart Jancar\surnameend},
  \bibinfo{author}{Frantisek \surnamestart Mr{\'{a}}z\surnameend} \&
  \bibinfo{author}{Martin \surnamestart Pl{\'{a}}tek\surnameend}
  (\bibinfo{year}{1993}): \emph{\bibinfo{title}{A Taxonomy of Forgetting
  Automata}}.
\newblock In: {\slshape \bibinfo{booktitle}{MFCS'93}}, {\slshape
  \bibinfo{series}{Lecture Notes in Computer Science}} \bibinfo{volume}{711},
  \bibinfo{publisher}{Springer}, pp. \bibinfo{pages}{527--536},
  \doi{10.1007/3-540-57182-5\_44}.

\bibitemdeclare{article}{La03}
\bibitem{La03}
\bibinfo{author}{Edmund \surnamestart Landau\surnameend}
  (\bibinfo{year}{1903}): \emph{\bibinfo{title}{{\"Uber} die Maximalordnung der
  Permutation gegebenen Grades}}.
\newblock {\slshape \bibinfo{journal}{Archiv der Mathematik und Physik}}
  \bibinfo{volume}{3}, pp. \bibinfo{pages}{92--103}.

\bibitemdeclare{book}{La09}
\bibitem{La09}
\bibinfo{author}{Edmund \surnamestart Landau\surnameend}
  (\bibinfo{year}{1909}): \emph{\bibinfo{title}{Handbuch der Lehre von der
  Verteilung der Primzahlen~I}}.
\newblock \bibinfo{publisher}{Teubner, Leipzig/Berlin}.

\bibitemdeclare{article}{MP01}
\bibitem{MP01}
\bibinfo{author}{Carlo \surnamestart Mereghetti\surnameend} \&
  \bibinfo{author}{Giovanni \surnamestart Pighizzini\surnameend}
  (\bibinfo{year}{2001}): \emph{\bibinfo{title}{Optimal Simulations between
  Unary Automata}}.
\newblock {\slshape \bibinfo{journal}{{SIAM} J. Comput.}}
  \bibinfo{volume}{30}(\bibinfo{number}{6}), pp. \bibinfo{pages}{1976--1992},
  \doi{10.1137/S009753979935431X}.

\bibitemdeclare{inproceedings}{MF71}
\bibitem{MF71}
\bibinfo{author}{Albert~R. \surnamestart Meyer\surnameend} \&
  \bibinfo{author}{Michael~J. \surnamestart Fischer\surnameend}
  (\bibinfo{year}{1971}): \emph{\bibinfo{title}{SWAT 1971}}.
\newblock \bibinfo{publisher}{{IEEE} Computer Society}, pp.
  \bibinfo{pages}{188--191}, \doi{10.1109/SWAT.1971.11}.

\bibitemdeclare{inproceedings}{Pig19}
\bibitem{Pig19}
\bibinfo{author}{Giovanni \surnamestart Pighizzini\surnameend}
  (\bibinfo{year}{2019}): \emph{\bibinfo{title}{Limited Automata: Properties,
  Complexity and Variants}}.
\newblock In: {\slshape \bibinfo{booktitle}{{DCFS} 2019}}, {\slshape
  \bibinfo{series}{Lecture Notes in Computer Science}} \bibinfo{volume}{11612},
  \bibinfo{publisher}{Springer}, pp. \bibinfo{pages}{57--73},
  \doi{10.1007/978-3-030-23247-4\_4}.

\bibitemdeclare{article}{PP14}
\bibitem{PP14}
\bibinfo{author}{Giovanni \surnamestart Pighizzini\surnameend} \&
  \bibinfo{author}{Andrea \surnamestart Pisoni\surnameend}
  (\bibinfo{year}{2014}): \emph{\bibinfo{title}{Limited Automata and Regular
  Languages}}.
\newblock {\slshape \bibinfo{journal}{Int. J. Found. Comput. Sci.}}
  \bibinfo{volume}{25}(\bibinfo{number}{7}), pp. \bibinfo{pages}{897--916},
  \doi{10.1142/S0129054114400140}.

\bibitemdeclare{article}{PP15}
\bibitem{PP15}
\bibinfo{author}{Giovanni \surnamestart Pighizzini\surnameend} \&
  \bibinfo{author}{Andrea \surnamestart Pisoni\surnameend}
  (\bibinfo{year}{2015}): \emph{\bibinfo{title}{Limited Automata and
  Context-Free Languages}}.
\newblock {\slshape \bibinfo{journal}{Fundam. Inform.}}
  \bibinfo{volume}{136}(\bibinfo{number}{1-2}), pp. \bibinfo{pages}{157--176},
  \doi{10.3233/FI-2015-1148}.

\bibitemdeclare{article}{PP19}
\bibitem{PP19}
\bibinfo{author}{Giovanni \surnamestart Pighizzini\surnameend} \&
  \bibinfo{author}{Luca \surnamestart Prigioniero\surnameend}
  (\bibinfo{year}{2019}): \emph{\bibinfo{title}{Limited automata and unary
  languages}}.
\newblock {\slshape \bibinfo{journal}{Inf. Comput.}} \bibinfo{volume}{266}, pp.
  \bibinfo{pages}{60--74}, \doi{10.1016/j.ic.2019.01.002}.

\bibitemdeclare{inproceedings}{PP23b}
\bibitem{PP23b}
\bibinfo{author}{Giovanni \surnamestart Pighizzini\surnameend} \&
  \bibinfo{author}{Luca \surnamestart Prigioniero\surnameend}
  (\bibinfo{year}{2023}): \emph{\bibinfo{title}{Once-Marking and Always-Marking
  $1$-Limited Automata}}.
\newblock In: {\slshape \bibinfo{booktitle}{AFL 2023}}, {\slshape
  \bibinfo{series}{Electronic Proceedings in Theoretical Computer Science}}
  \bibinfo{volume}{386}, pp. \bibinfo{pages}{215--227},
  \doi{10.4204/EPTCS.386.17}.

\bibitemdeclare{inproceedings}{PP23}
\bibitem{PP23}
\bibinfo{author}{Giovanni \surnamestart Pighizzini\surnameend} \&
  \bibinfo{author}{Luca \surnamestart Prigioniero\surnameend}
  (\bibinfo{year}{2023}): \emph{\bibinfo{title}{Two-way Machines and de
  {B}ruijn Words}}.
\newblock In: {\slshape \bibinfo{booktitle}{CIAA 2023}}, {\slshape
  \bibinfo{series}{Lecture Notes in Computer Science}} \bibinfo{volume}{14151},
  pp. \bibinfo{pages}{254--65}, \doi{10.1007/978-3-031-40247-0_19}.

\bibitemdeclare{article}{PPS22}
\bibitem{PPS22}
\bibinfo{author}{Giovanni \surnamestart Pighizzini\surnameend},
  \bibinfo{author}{Luca \surnamestart Prigioniero\surnameend} \&
  \bibinfo{author}{Simon \surnamestart {\v{S}}{\'{a}}dovsk{\'{y}}\surnameend}
  (\bibinfo{year}{2022}): \emph{\bibinfo{title}{1-Limited Automata: Witness
  Languages and Techniques}}.
\newblock {\slshape \bibinfo{journal}{J. Autom. Lang. Comb.}}
  \bibinfo{volume}{27}(\bibinfo{number}{1-3}), pp. \bibinfo{pages}{229--244},
  \doi{10.25596/jalc-2022-229}.

\bibitemdeclare{book}{Sha08}
\bibitem{Sha08}
\bibinfo{author}{Jeffrey~O. \surnamestart Shallit\surnameend}
  (\bibinfo{year}{2008}): \emph{\bibinfo{title}{A Second Course in Formal
  Languages and Automata Theory}}.
\newblock \bibinfo{publisher}{Cambridge University Press},
  \doi{10.1017/CBO9780511808876}.

\bibitemdeclare{article}{She59}
\bibitem{She59}
\bibinfo{author}{John~C. \surnamestart Shepherdson\surnameend}
  (\bibinfo{year}{1959}): \emph{\bibinfo{title}{The Reduction of Two-Way
  Automata to One-Way Automata}}.
\newblock {\slshape \bibinfo{journal}{{IBM} J. Res. Dev.}}
  \bibinfo{volume}{3}(\bibinfo{number}{2}), pp. \bibinfo{pages}{198--200},
  \doi{10.1147/rd.32.0198}.

\bibitemdeclare{article}{Sip80b}
\bibitem{Sip80b}
\bibinfo{author}{Michael \surnamestart Sipser\surnameend}
  (\bibinfo{year}{1980}): \emph{\bibinfo{title}{Lower Bounds on the Size of
  Sweeping Automata}}.
\newblock {\slshape \bibinfo{journal}{J. Comput. Syst. Sci.}}
  \bibinfo{volume}{21}(\bibinfo{number}{2}), pp. \bibinfo{pages}{195--202},
  \doi{10.1016/0022-0000(80)90034-3}.

\bibitemdeclare{article}{Sz80}
\bibitem{Sz80}
\bibinfo{author}{Mih\'aly \surnamestart Szalay\surnameend}
  (\bibinfo{year}{1980}): \emph{\bibinfo{title}{On the maximal order in {$S_n$}
  and~{$S_n^*$}}}.
\newblock {\slshape \bibinfo{journal}{Acta Arithmetica}} \bibinfo{volume}{37},
  pp. \bibinfo{pages}{321--331}, \doi{10.4064/aa-37-1-321-331}.

\bibitemdeclare{book}{WW86}
\bibitem{WW86}
\bibinfo{author}{Klaus~W. \surnamestart Wagner\surnameend} \&
  \bibinfo{author}{Gerd \surnamestart Wechsung\surnameend}
  (\bibinfo{year}{1986}): \emph{\bibinfo{title}{Computational complexity}}.
\newblock \bibinfo{publisher}{D.~Reidel Publishing Company, Dordrecht}.

\end{thebibliography}

\end{document}